\documentclass[a4paper,USenglish, cleveref, autoref, thm-restate]{lipics-v2021}

\pdfoutput=1

\hideLIPIcs

\nolinenumbers

\usepackage{tikz}
\usepackage{enumitem}
\usepackage[utf8]{inputenc}
\usepackage[T1]{fontenc}

\usepackage{amsmath, amssymb}
\usepackage{caption}
\usepackage{subcaption}

\usepackage{hyperref}

\newcommand{\NN}{\mathbb{N}} 

\newcommand{\tp}{\mathsf{TreePartitioning}}
\newcommand{\neighbors}{\mathsf{Neighbors}}
\newcommand{\pieces}{\mathsf{Pieces}}
\newcommand{\Pp}{\mathcal{P}}
\newcommand{\Id}{\mathrm{Id}}

\tikzset{every picture/.style={line width=0.75pt}} 

\title{A subquadratic certification scheme for $P_5$-free graphs}

\author{Nicolas Bousquet}{CNRS, INSA Lyon, UCBL, LIRIS, UMR5205, F-69622 Villeurbanne, France \and CNRS - Université de Montréal CRM – CNRS, Canada}{}{}{}{}

\author{Sébastien Zeitoun}{CNRS, INSA Lyon, UCBL, LIRIS, UMR5205, F-69622 Villeurbanne, France}{}{}{}{}

\authorrunning{N. Bousquet and S. Zeitoun}

\ccsdesc{Theory of computation~Distributed algorithms}

\begin{document}
	\maketitle
	
	\section{Introduction}
	
	
	In distributed computing, designing efficient algorithms to detect substructures in a network is a research field that has been very active recently.
	Several variants of the problem (detection, counting, listing, testing) have been studied in different models (particularly in \textsc{Congested Clique} and \textsc{Congest}). We refer to~\cite{CensorHillel22} for a recent survey on this topic.
	
	In this work, we study the subgraph detection problem in the model of \emph{local certification}. In local certification, every node of a network receives a \emph{certificate}, given by an external prover, whose aim is to convince all the nodes that some given property $\Pp$ is satisfied by the network.
	Each node outputs a binary decision (\emph{accept} or \emph{reject}), based only on its neighborhood (that is: the unique identifiers of its neighbors and their certificates\footnote{We will mention more general result where a node takes a decision based on its neighborhood at distance at most $d$ instead of $1$. The integer $d$ is called the \emph{verification radius} and when no verification radius is mentioned, we will assume that $d=1$.}).
	A certification scheme is said to be correct if it fulfills the \emph{completeness} and \emph{soundness} conditions, namely: if $\Pp$ is satisfied, the prover should be able to attribute certificates in such a way that all the vertices accept (completeness), and if $\Pp$ is not satisfied, for every assignment of the certificates by the prover, at least one node should reject (soundness). See~\cite{Feuilloley21} for an introduction to local certification.
	
	We usually measure the quality of a certification scheme by looking at the size of the certificates (which is usually expressed as a function of $n$, the number of vertices in the network): the smaller the certificates are, the better the scheme is. One of the fundamental results in certification is that any property can be certified with certificates of size $O(n^2)$~\cite{GoosS16} (the idea being roughly that the whole adjacency matrix of the graph can be given as a certificate to every vertex). This quadratic upper bound is known to be essentially tight for some properties such as graph isomorphism~\cite{GoosS16}. Nevertheless, some properties admit a much more compact certification, with certificates of poly-logarithmic size. An important line of research has been dedicated to the design of compact certification schemes, see e.g. the recent meta-theorems in that direction~\cite{FeuilloleyBP22,FraigniaudM0RT23,FraigniaudMRT24}.
	
	The existence of a substructure can be certified with compact certificates (i.e. of poly-logarithmic size). Indeed, the certification consists in encoding in the certificates which vertices belong to this substructure (see~\cite{BousquetCFPZ24+} for a more detailed discussion).
	Certifying the non-existence of a substructure turns out to be much harder. When this forbidden structure is a (collection of) minors, polylogarithmic size certificates have been proven to be sufficient in many cases such as for planar graphs~\cite{FeuilloleyFMRRT21}, bounded genus graphs~\cite{EsperetL22,FraigniaudM0RT23},
	small excluded minors~\cite{BousquetFP21,FeuilloleyBP22}, paths and cycle excluded minors~\cite{FeuilloleyBP22}, planar excluded minors~\cite{FraigniaudMRT24}. 
	
	In this paper, we focus on the absence of an induced subgraph~$H$ (we will also call this property \emph{$H$-freeness}). In that case, the property becomes harder to certify.
	Let us start with a simple case. Imagine that we want to certify that a graph does not contain some clique $K_k$ on $k$ vertices.
	An easy way to do it is to write in the certificate of every vertex its list of neighbors (each vertex can easily check if this information is correct).
	By doing so, each vertex is able to know the adjacency of its neighbors by looking at their certificates, and can thus determine the size of the largest clique to which it belongs (and it rejects if it is at least~$k$).
	However, this certification scheme has certificates of size $O(n \log n)$, so certificates of polynomial size. In fact, it has been proven in~\cite{BousquetEFZ24} that we can reduce this size to $O(n)$ by using a renaming argument. On the other hand, a construction of~\cite{DruckerKO13} that can be easily adapted to local certification leads to a $\Omega(n)$ lower bound if $k \geqslant 4$, and $\Omega(n/e^{O(\sqrt{\log n})})$ if $k=3$. So the following holds:
	
	\begin{theorem}[\cite{BousquetEFZ24}, \cite{DruckerKO13}]
		\label{thm:cliques}
		The optimal size of the certificates for $K_k$-freness is $\Theta(n)$ if $k\geqslant 4$, and lies between $\Omega(n/e^{O(\sqrt{\log n})})$ and $O(n)$ if $k=3$.
	\end{theorem}
	
	Theorem~\ref{thm:cliques} ensures that when $H$ is a clique, polynomial size certificates are required for the certification of $H$-freeness, even when $H$ is a triangle. When $H$ is sparse, it turns out that certifying $H$-free graphs is not easier. The certification of $H$-freeness, where $H$ is a large enough path or cycle, needs certificates of size at least linear in $n$~\cite{BousquetCFPZ24+}. They also proposed superlinear but subquadratic upper bounds when the verification radius is at least $2$ (in their model, to perform the verification, vertices do not only see their neighbors but also all the vertices and edges up to some constant distance). These results, together with Theorem~\ref{thm:cliques}, suggest that the certification of $H$-freeness can not be achieved in general with compact certificates.
	
	\paragraph*{Certification of $P_k$-free graphs.}
	
	In this work, we focus on the certification of the $P_k$-freeness property, where $P_k$ denotes the path having $k$ vertices.
	The case $k=3$ is easy to handle: indeed, connected $P_3$-free graphs are cliques, and it is of common knowledge that $O(\log n)$ bits are sufficient to certify that a graph is a clique (for instance, one just needs to write the number of vertices $n$ in the certificates of all the vertices, certify it using a spanning-tree, and then each vertex can check that its degree is $n-1$; see~\cite{Feuilloley21} for more details on the use of spanning-trees). For $k=4$, certificates of size $O(\log n)$ are still sufficient~\cite{FraigniaudM0RT23}. Indeed, the class of $P_4$-free graphs, also known as \emph{cographs}, has a very specific structure that can be exploited to design an efficient certification scheme.
	For $k \geqslant 5$, to the best of our knowledge, no certification scheme (with verification radius $1$) using $o(n^2)$ bits was known prior to this paper.
	
	In~\cite{BousquetCFPZ24+}, authors studied the certification of $P_k$-free graphs in a more general setting where the vertices have verification radius $d \geqslant 1$ (recall that it means that they can see the graph up to some constant distance~$d$ and that our model corresponds to the case $d=1$). They provided some non-trivial upper and lower bounds on the size of certificates.
	In terms of lower-bounds, they proved the following:
	
	\begin{theorem}[\cite{BousquetCFPZ24+}]
		\label{thm:lower bound}
		In the local certification model where the nodes have a verification radius $d \geqslant 1$, $\Omega_d(n)$ bits are necessary to certify $P_{4d+3}$-free graphs.
	\end{theorem}
	
	Theorem~\ref{thm:lower bound} ensures that, if paths are long enough compared to the verification radius, certifying $P_k$-free graphs requires certificates of polynomial size. When $d=1$, Theorem~\ref{thm:lower bound} ensures that certificates of size $\Omega(n)$ are necessary to certify $P_7$-free graphs. However, Theorem~\ref{thm:lower bound} does not apply for $k \in \{5,6\}$. For $k \in \{5,6\}$, we are aware of an unpublished slightly worse but still polynomial lower bound~\cite{ChaniotisCHS24}: namely,~$\Omega(n/e^{O(\sqrt{\log n})})$. The idea of the proof is to tweak the lower-bound of the triangle-free construction in~\cite{DruckerKO13} and to adapt it for paths (roughly, by replacing some edges by non-edges and adding some vertices).

	In terms of upper bounds, authors in~\cite{BousquetCFPZ24+} provide a generic technique to obtain subquadratic upper bounds. However, this technique only works when the verification radius of the vertices is at least~$2$. These upper bounds are the following:
	
	\begin{theorem}[\cite{BousquetCFPZ24+}]
		\label{thm:P14/3d}
		In the local certification model where the nodes have a verification radius $d \geqslant 2$, the following upper bounds hold:
		\begin{itemize}
			\item $O(n\log^3 n)$ bits are sufficient to certify $P_k$-free graphs for all $k \leqslant 3d-1$;
			\item $O(n^{3/2}\log^2 n)$ bits are sufficient to certify $P_k$-free graphs for all $k \leqslant \left\lceil \frac{14}{3}d\right\rceil-1$.
		\end{itemize}
	\end{theorem}
	
	The proof technique of Theorem~\ref{thm:P14/3d} heavily relies on the fact that $d \geqslant 2$. 
	Indeed, one of the cornerstones of the proof is to cut the adjacency matrix of the graph in small pieces, and spread these pieces in the certificates of the vertices such that the big-degree vertices are able to see all the pieces in the certificates of their neighbors and reconstruct the graph. But with such a certification, to guarantee to the vertices that the graph they reconstruct is correct, they need to check that they reconstructed the same graph as their neighbors, and the verification radius should be at least~$2$ to make this verification possible.
	In our model where the verification radius is~$1$, such an approach is very difficult to adapt.
	
	Authors in~\cite{BousquetCFPZ24+} conjectured that subquadratic upper bounds exist for any length of paths. More precisely, they conjectured the following holds:
	
	\begin{conjecture}[\cite{BousquetCFPZ24+}]
		\label{conj:p_alpha*d subquadratic}
		For all $\alpha > 0$, the optimal size for the local certification of $P_{\alpha d}$-free graphs with verification radius $d$ is of the form $\Theta(n^{2-1/f(\alpha)})$, for some unbounded increasing function~$f$.
	\end{conjecture}
	
	The aim of this work is to give the first subquadratic upper bound for $P_5$-free graphs, in the model where the verification radius is~1. Namely, we prove that the following holds:
	
	\begin{restatable}{theorem}{ThmUpperBound}
		\label{thm:p5_free}
		There is a certification scheme for $P_5$-free graphs with certificates of size $O(n^{3/2})$.
	\end{restatable}
	
	Theorem~\ref{thm:p5_free} makes a first step towards Conjecture~\ref{conj:p_alpha*d subquadratic}, in the case where $d=1$. As we already explained, the proof technique of~\cite{BousquetCFPZ24+} cannot be adapted when $d=1$ and we have to adopt a new strategy. More concretely, our proof relies on the following structural theorem proven by Basc\'o and Tuza~\cite{bacso1990dominating}:
	
	\begin{theorem}[\cite{bacso1990dominating}]
		\label{thm:dominating p3 or clique}
		Every connected $P_5$-free graph has a dominating set which is either a clique, or an induced $P_3$.
	\end{theorem}
	
	Theorem~\ref{thm:dominating p3 or clique} will allow us to find some useful partition of the vertex set, which we will call a \emph{valid tree partition}. Basically, a valid tree partition of a graph~$G$ is a partition of its vertices, where the sets of the partition are the vertices of some tree. The root is a dominating set $X$ in~$G$, that is isomorphic to a clique or an induced~$P_3$. The children of the root are the tree-partitions of the connected components of $G \setminus X$. See Section~\ref{sec:tree_partitions} for a more formal definition.
	
	While the technique of spreading the adjacency matrix of the graph is possible if the dominating set $X$ is a large clique, it is not in the other cases. In these cases, we set the certificates of the vertices of $X$ to be the neighborhood of $X$. A careful analysis of this scheme permits to prove that we can indeed detect a $P_5$ if $G$ contains one. We refer the reader to Section~\ref{sec:proof} for the proof of Theorem~\ref{thm:p5_free}.
	
	\paragraph*{Open problems}
	A number of open questions about the certification of $P_k$-free graphs remain open.
	First, the lower and upper bounds for the certification of $P_k$-free graphs given by Theorems~\ref{thm:lower bound} and~\ref{thm:P14/3d} do not match. A natural, but probably hard, open problem is to close this gap. Also note that the upper bounds of Theorem~\ref{thm:P14/3d} hold only for $P_k$-free graphs if $k \leqslant \left\lceil \frac{14}{3}d\right\rceil-1$, where $d$ is the verification radius. Another interesting open question is to determine if it is possible to push further this upper bound for longer paths and prove Conjecture~\ref{conj:p_alpha*d subquadratic}.
	
	When the verification radius is~$1$, Theorem~\ref{thm:lower bound} gives a linear lower bound for the certification of $P_7$-free graphs and~\cite{ChaniotisCHS24} gives a $\Omega(n/e^{O(\sqrt{\log n})})$ lower bound for $P_5$-free graphs and $P_6$-free graphs. To our knowledge no superlinear lower bound is known for $P_k$-free graphs even for very large~$k$, and it would be interesting to prove one for some large enough~$k$. 
	
	In terms of upper bounds, we did not succeed to decrease further the upper bound of Theorem~\ref{thm:p5_free}, but we believe that our proof technique may be used as a basis to improve further the upper bound for $P_5$-free graphs.
	
	Note that Camby and Schaudt actually proved in~\cite{CambyS16} a generalization of Theorem~\ref{thm:dominating p3 or clique} that holds for $P_k$-free graphs in general (namely, they proved that every connected $P_k$-free graph has a connected dominating set which is either $P_{k-2}$-free, or isomorphic to $P_{k-2}$). Theorem~\ref{thm:dominating p3 or clique} corresponds to the case $k=5$, since every connected $P_3$-free graph is a clique. However, for larger values of $k$, this dominating clique can be replaced with an arbitrarily large set of low connectivity and we did not succeed to use it to design small certification schemes.
	Finding a subquadratic upper bound for $P_6$-free graphs with verification radius $1$ probably needs new ideas and is an interesting open problem.


	\section{Model and definitions}
	
	\subsection{Graphs}
	We model the network by a connected undirected graph, without loops and multiple edges. We denote by~$V(G)$ the set of vertices of a graph $G$, by~$E(G)$ its set of edges, and we will often denote by $n$ its number of vertices. For every positive integer~$k$, we denote by $P_k$ the graph which is a path on~$k$ vertices (and has $k-1$ edges), that is, the graph on vertex set $\{1, \ldots, k\}$ and edge set all the pairs $\{i,j\}$ such that $|i-j|=1$. If $G,H$ are two graphs, we say that \emph{$H$ is an induced subgraph of~$G$} if $H$ can be obtained from $G$ by deleting a subset of vertices $S \subseteq V(G)$ as well as all the edges incident to some vertex in $S$.
	
	\subsection{Local certification}
	
	In local certification, we consider $n$-vertices graphs where each vertex is equipped with a \emph{unique identifier}, which is an integer in the range~$\{1, \ldots, n^c\}$ for some constant $c>0$. More precisely: an \emph{identifier assignment} for a graph~$G$ is an injective function $\Id : V(G) \rightarrow \{1, \ldots, n^c\}$. Note that each unique identifier can be encoded on $O(\log n)$ bits.
	
	Let $\Pp$ be a property on graphs and $C$ be a set (called the set of \emph{certificates}). A \emph{certificate assignment} is a function~$c : V(G) \rightarrow C$. A \emph{local verification procedure} is an algorithm that outputs a binary decision (\emph{accept} or~\emph{reject}) and takes as input the closed neighborhood of a vertex~$u$, that is:
	\begin{itemize}
		\item the unique identifier of~$\Id(u)$ of~$u$,
		\item the certificate~$c(u)$ of~$u$,
		\item the set of identifiers and certificates of its neighbors $\{(\Id(v), c(v)) \; | \{u,v\} \in E(G)\}$
	\end{itemize}
	
	Let $s: \NN \rightarrow \NN$. We say that there exists a~\emph{local certification scheme for~$\Pp$ with certificates of size~$s$} if there exists a local verification procedure such that, for every $n$-vertex graph~$G$ equipped with an identifier assignment~$\Id$:
	\begin{itemize}
		\item if $G$ satisfies~$\Pp$, then there exists a certificate assignment $c: V(G) \rightarrow \{1, \ldots, 2^{s(n)}\}$ such that all the vertices accept (completeness)
		\item if $G$ does not satisfy~$\Pp$, then for all certificate assignment $c: V(G) \rightarrow \{1, \ldots, 2^{s(n)}\}$, there is at least one vertex which rejects (soundness)
	\end{itemize}
	
	We will often say that the certificates are assigned by an external~\emph{prover}, who tries to convince the vertices that~$\Pp$ is satisfied, but who should be able to succeed if and only if~$\Pp$ is indeed satisfied.
	
	Let us make a last remark about identifiers. In the following, we will give a polynomial upper bound for the certification of $P_5$-free graphs. In~\cite{BousquetEFZ24}, authors proved a renaming theorem: more precisely, they showed that, at a cost of $O(\log n)$ bits in the certificates, we can assume that the identifiers of the vertices are in~$\{1, \ldots, n\}$. Thus, we will make this assumption in the rest of this note, and this will not change the upper bound of Theorem~\ref{thm:p5_free} since $O(\log n) = o(n^{3/2})$.
	Similarly, using a technique based on spanning-trees and at a cost of $O(\log n)$ bits in the certificates, we can assume that all the vertices know~$n$ (see \cite{Feuilloley21} for details). Thus, we make this assumption too.

	\section{Valid tree partitions}
	\label{sec:tree_partitions}
	
	A \emph{rooted tree} is a pair $(T,r)$ where $T$ is a tree and $r \in V(T)$. Let $u,v,w \in V(T)$. Since $T$ is a tree, there is a unique path from~$r$ to~$v$. We say that \emph{$u$ is an ancestor of $v$} (or equivalently that \emph{$v$ is a descendant of~$u$}) if $u$ belongs to the path from $r$ to $v$. We say that \emph{$u$ is the parent of~$v$} (or equivalently that \emph{$v$ is a child of~$u$}) if $u$ is an ancestor of~$v$ and there is an edge between~$u$ and~$v$. Finally, we say that $u$ is the \emph{closest common ancestor of~$v$ and~$w$} if $u$ is an ancestor of both~$v$ and~$w$, and if every other common ancestor~$u'$ of both~$v$ and~$w$ is also an ancestor of~$u$.

	
	

	Let $G$ be a graph. We say that a pair $(T, \beta)$ is a \emph{tree partition} of $G$ if $T$ is a rooted tree and $\beta$ is a function \mbox{$V(T) \rightarrow 2^{V(G)}$} which forms a partition of $V(G)$. In other words, for each $s \in V(T)$, $\beta(s) \neq \emptyset$, and for each $v \in V(G)$, there exists a unique $s \in V(T)$ such that $v \in \beta(s)$. For each $s \in V(T)$, let us denote by $T_s$ the subtree of $T$ rooted in $s$ (that is the subtree of $T$ consisting of $s$ and all the descendants), and by $G_s$ the subgraph of $G$ induced by $\beta(T_s)$.
	
	\begin{definition}
		\label{def:tree_partitioning}
		We say that a tree partition $(T, \beta)$ is a \emph{valid tree-partition} of $G$ if, for every $s \in V(T)$, the following conditions are satisfied: 
		\begin{enumerate}[label=(\theenumi)]
			\item\label{defi} $\beta(s)$ is either a clique or an induced $P_3$ in $G$;
			\item\label{defii} $\beta(s)$ is a dominating set in $G_s$;
			\item\label{defiii} if $s$ is not a leaf of $T$ and has $s_1, \ldots, s_k$ as children, the connected components of $G_s \setminus \beta(s)$ are exactly $G_{s_1}, \ldots, G_{s_k}$.
		\end{enumerate}
	\end{definition}
	
	We can now remark that every $P_5$-free graph $G$ has at least one valid tree-partition, which can be obtained as follows. We start with an empty tree $T$. By Theorem~\ref{thm:dominating p3 or clique}, there exists a subset $S$ of vertices which is a dominating clique or a dominating $P_3$. We create a new node at the root $r$ of the tree and set $\beta(r)=S$. Then, we apply this procedure inductively on all the connected components of $G \setminus S$ (each of them creating a new subtree) where roots of these components are children of $r$. The three points of the definition easily follow from the construction and the definition of $S$.
	Note $G$ might have several valid tree-partitions.
	
	Let us prove the following basic property of tree partitions:
	
	\begin{lemma}
		\label{claim:edges_ancestors}
		Let $G=(V,E)$ be a graph having a valid tree-partition $(T, \beta)$, and let $\{u,v\} \in E$. Let $s, t \in V(T)$ be such that $u \in \beta(s)$ and $v \in \beta(t)$. Then, in $T$, either $s$ is an ancestor of $t$, or $t$ is an ancestor of $s$.
	\end{lemma}
	\begin{proof}
		This statement follows from~\ref{defiii} of Definition~\ref{def:tree_partitioning}. If neither $s$ nor $t$ is an ancestor of the other, let $q$ be their closest common ancestor. By Definition~\ref{def:tree_partitioning}\ref{defiii}, $u$ and $v$ belong to two different connected components in $G_q \setminus \beta(q)$, and in particular are not adjacent.
	\end{proof}

	\section{Proof of Theorem~\ref{thm:p5_free}.}
	\label{sec:proof}
	
	The remaining of this paper is devoted to the proof of~Theorem~\ref{thm:p5_free}:
	\ThmUpperBound*
	
	Let us describe a certification scheme for $P_5$-free graphs using certificates of size $O(n^{3/2})$.
	Let $G$ be a $n$-vertex graph with unique identifiers in $\{1, \ldots, n\}$.
	We describe the certification, then explain how nodes check their certificate and finally prove the correctness of the construction.

	\paragraph*{Certification.}
	If $G$ is $P_5$-free, the certificates given by the prover to the vertices of $G$ consist in three parts. For every vertex $u \in V$, we will denote these parts by $\neighbors(u)$, $\tp(u)$, $\pieces(u)$. \medskip

	\noindent
	First, in $\neighbors(u)$, the prover writes the list of neighbors of $u$. Since identifiers of the vertices are in $\{1, \ldots, n\}$, this can be done with $O(n)$ bits, by writing a $n$-bit vector $T[u]$ such that $T[u]_i$ is equal to $1$ if and only if $u$ is a neighbor of the vertex having identifier~$i$.
	\medskip
	
	\noindent
	In $\tp(u)$, the prover writes a valid tree-partition $(T, \beta)$ of $G$, as follows. First, it gives the structure of the rooted tree~$T$. Since~$T$ has at most $n$ nodes, it is of common knowledge that it can be done with $O(n)$ bits, for instance in the following way: by drawing the tree with the root on the top and a bottom-up child relationship, following the edges in a depth-first search and writing a~$0$ for each downward edge and a~$1$ for each upward edge (see Figure~\ref{fig:example_tree_encoding} for an example).
	Then, for each $s \in V(T)$, the prover writes which vertices are in $\beta(s)$, and if $\beta(s)$ is a clique or an induced $P_3$ in $G$. Finally, if $\beta(s)$ is a $P_3$, the prover writes which vertex is the center of the $P_3$ in the certificate.
	In total, this part of the certificate has size $O(n \log n)$.
	
	\begin{figure}[h]
		\centering
		
		\begin{tikzpicture}[x=0.75pt,y=0.75pt,yscale=-1,xscale=1]
			
			\draw    (246.5,199) -- (265.5,149.5) ;
			\draw    (265.5,149.5) -- (284.5,199) ;
			\draw    (284.5,100) -- (303.5,149.5) ;
			\draw    (265.5,149.5) -- (284.5,100) ;
			\draw    (222.5,100) -- (222.5,149.5) ;
			\draw    (253.5,50.5) -- (222.5,100) ;
			\draw    (253.5,50.5) -- (253.5,100) ;
			\draw    (253.5,50.5) -- (284.5,100) ;
			\draw  [fill={rgb, 255:red, 255; green, 255; blue, 255 }  ,fill opacity=1 ] (247,50.5) .. controls (247,46.91) and (249.91,44) .. (253.5,44) .. controls (257.09,44) and (260,46.91) .. (260,50.5) .. controls (260,54.09) and (257.09,57) .. (253.5,57) .. controls (249.91,57) and (247,54.09) .. (247,50.5) -- cycle ;
			\draw  [fill={rgb, 255:red, 255; green, 255; blue, 255 }  ,fill opacity=1 ] (216,100) .. controls (216,96.41) and (218.91,93.5) .. (222.5,93.5) .. controls (226.09,93.5) and (229,96.41) .. (229,100) .. controls (229,103.59) and (226.09,106.5) .. (222.5,106.5) .. controls (218.91,106.5) and (216,103.59) .. (216,100) -- cycle ;
			\draw  [fill={rgb, 255:red, 255; green, 255; blue, 255 }  ,fill opacity=1 ] (247,100) .. controls (247,96.41) and (249.91,93.5) .. (253.5,93.5) .. controls (257.09,93.5) and (260,96.41) .. (260,100) .. controls (260,103.59) and (257.09,106.5) .. (253.5,106.5) .. controls (249.91,106.5) and (247,103.59) .. (247,100) -- cycle ;
			\draw  [fill={rgb, 255:red, 255; green, 255; blue, 255 }  ,fill opacity=1 ] (278,100) .. controls (278,96.41) and (280.91,93.5) .. (284.5,93.5) .. controls (288.09,93.5) and (291,96.41) .. (291,100) .. controls (291,103.59) and (288.09,106.5) .. (284.5,106.5) .. controls (280.91,106.5) and (278,103.59) .. (278,100) -- cycle ;
			\draw  [fill={rgb, 255:red, 255; green, 255; blue, 255 }  ,fill opacity=1 ] (216,149.5) .. controls (216,145.91) and (218.91,143) .. (222.5,143) .. controls (226.09,143) and (229,145.91) .. (229,149.5) .. controls (229,153.09) and (226.09,156) .. (222.5,156) .. controls (218.91,156) and (216,153.09) .. (216,149.5) -- cycle ;
			\draw  [fill={rgb, 255:red, 255; green, 255; blue, 255 }  ,fill opacity=1 ] (297,149.5) .. controls (297,145.91) and (299.91,143) .. (303.5,143) .. controls (307.09,143) and (310,145.91) .. (310,149.5) .. controls (310,153.09) and (307.09,156) .. (303.5,156) .. controls (299.91,156) and (297,153.09) .. (297,149.5) -- cycle ;
			\draw  [fill={rgb, 255:red, 255; green, 255; blue, 255 }  ,fill opacity=1 ] (259,149.5) .. controls (259,145.91) and (261.91,143) .. (265.5,143) .. controls (269.09,143) and (272,145.91) .. (272,149.5) .. controls (272,153.09) and (269.09,156) .. (265.5,156) .. controls (261.91,156) and (259,153.09) .. (259,149.5) -- cycle ;
			\draw  [fill={rgb, 255:red, 255; green, 255; blue, 255 }  ,fill opacity=1 ] (240,199) .. controls (240,195.41) and (242.91,192.5) .. (246.5,192.5) .. controls (250.09,192.5) and (253,195.41) .. (253,199) .. controls (253,202.59) and (250.09,205.5) .. (246.5,205.5) .. controls (242.91,205.5) and (240,202.59) .. (240,199) -- cycle ;
			\draw  [fill={rgb, 255:red, 255; green, 255; blue, 255 }  ,fill opacity=1 ] (278,199) .. controls (278,195.41) and (280.91,192.5) .. (284.5,192.5) .. controls (288.09,192.5) and (291,195.41) .. (291,199) .. controls (291,202.59) and (288.09,205.5) .. (284.5,205.5) .. controls (280.91,205.5) and (278,202.59) .. (278,199) -- cycle ;
			
		\end{tikzpicture}
		
		\caption{A rooted tree whose encoding is 0011010001011011.}
		\label{fig:example_tree_encoding}
	\end{figure}
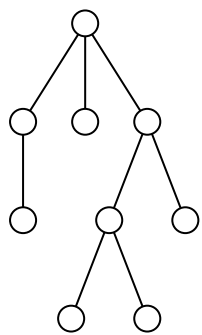

	\medskip
	
	\noindent
	In $\pieces(u)$, the prover writes the following. Let $s \in V(T)$ be the unique vertex of $T$ such that $u \in \beta(s)$. There are three cases:
	\begin{itemize}
		\item If $\beta(s)$ is an induced $P_3$ or a clique with at most $\sqrt{n}$ vertices, the prover writes in $\pieces(u)$ the neighborhood of all the vertices that belong to $\beta(s)$. As for $\neighbors(u)$, encoding the neighborhood of every vertex $v$ can be done with $O(n)$ bits, by coding its identifier on $O(\log n)$ bits and then writing a $n$-bit vector $T[v]$ such that $T[v]_i$ is equal to~$1$ if and only if $v$ is a neighbor of the vertex having identifier~$i$. Thus, in that case, $\pieces(u)$ has size $O(n^{3/2})$.
		\item If $\beta(s)$ is a clique on at least $\sqrt{n}$ vertices, the prover writes in $\pieces(u)$ the neighborhood of $u$ plus $\sqrt{n}$ other vertices of $G_s$. It does it in such a way, for every $v \in G_s$, there exists $w \in \beta(s)$ such that the neighborhood of $v$ is written in $\pieces(w)$. Note that it is possible, because the size of the clique is at least $\sqrt{n}$.
		Again, encoding the neighborhood of a single vertex can be done with $O(n)$ bits, so in total $\pieces(u)$ has size $O(n^{3/2})$ too.
	\end{itemize}
	The overall size of the certificates is thus $O(n^{3/2})$.

	\paragraph*{Verification.}
	Let us now explain how each vertex $u \in V$ checks the validity of its certificate.
	\begin{enumerate}[label=(\roman*)]
		\item First, $u$ checks its list of neighbors and rejects if $\neighbors(u)$ is not correct.
		
		\item Then, the node $u$ checks if $\tp(u)=\tp(v)$ for every neighbor $v$, and if every vertex appears in exactly one subset; $u$ rejects if it is not the case.
		
		If no vertex has rejected at this point, the prover has written the same tree-partition $(T, \beta)$ in all the certificates.
		
		\item The goal of the next step is to check that $(T, \beta)$ is a valid tree partition. Let $s \in V(T)$ be such that $u \in \beta(s)$.
		If $\beta(s)$ is supposed to be a clique, $u$ checks that it is indeed adjacent to all the vertices in $\beta(s)$. \\
		If $\beta(s)$ is supposed to be a $P_3$ and $u$ is the center (resp. an endpoint), it checks if it is adjacent to both endpoints (resp. to the central vertex, and not to the other endpoint). \\
		Then, for every $t \in V(T)$ which is an ancestor of $s$, $u$ checks if it is adjacent to at least one vertex in~$\beta(t)$. \\
		Finally, $u$ rejects if there exists $t \in V(T)$ which is an ancestor of $s$, having at least two children $t_1$ and $t_2$, such that $u$ has a neighbor $v$ with $u \in G_{t_1}$ and $v \in G_{t_2}$.

		\item Let $s \in V(T)$ be such that $u \in \beta(s)$. 
		The node $u$ performs the following checks and rejects if one of them is incorrect. \\
		If $\beta(s)$ is a $P_3$ or a clique of size at most $\sqrt{n}$, then, for every $v \in \beta(s)$, $u$ checks if $\pieces(u)=\pieces(v)$, and if its own neighborhood is correctly written in $\pieces(u)$. \\
		If $\beta(s)$ is a clique of size at least $\sqrt{n}$, $u$ checks that, for every $v \in G_s$, there exists $w \in \beta(s)$ such that the neighborhood of $v$ is written in $\pieces(w)$. Moreover, if $v$ is adjacent to $u$ then $u$ checks that the neighborhood of $v$ in $\pieces(w)$ and in $\neighbors(v)$ are the same. 
		
		\item The last step of the verification consists in trying to determine if there is a $P_5$. At this point, the vertex $u$ (is supposed to) know its neighbors, the neighborhoods of its neighbors~$v$, the neighborhoods of some other vertices in the graph (written in $\pieces(v)$ for every neighbor $v$), and the valid tree-partition $(T, \beta)$ of $G$. In particular, if two vertices $v$ and $w$ are in $\beta(s)$ for some $s \in V(T)$, $u$ knows if $v$ and $w$ are adjacent or not. Similarly, if two vertices $v$ and $w$ are in different \emph{branches} of $T$ (in other words, if there exist $s, t \in V(T)$ which are not ancestors one of each other such that $v \in \beta(s)$ and $w \in \beta(t)$), $u$ knows that $v$ and $w$ are not adjacent. If, with all this information, $u$ detects an induced $P_5$ (which means that with the rules above it knows for sure all the edges and non-edges), then $u$ rejects.
		
		\item If $u$ did not reject previously, it accepts.
	\end{enumerate}

	\paragraph*{Correctness.} If $G$ is a $P_5$-free graph, if the prover gives the certificates described above, every vertex accepts. Let us show that, conversely, if $G$ has an induced $P_5$, at least one vertex rejects.
	
	If no vertex rejects in step (i), then~$\neighbors(u)$ is correct for every vertex~$u$. Thus, every neighbor~$v$ adjacent to~$u$ knows the adjacency of~$u$.
	Moreover, if no vertex rejects in steps (ii) and (iii), then the tree-partition $(T,\beta)$ written in the certificates is valid, because vertices checked the three properties of Definition~\ref{def:tree_partitioning}. Note that, for property~(3), the vertices just checked that for every $s \in V(T)$ having as children $s_1, \ldots, s_k$ in $T$, the induced subgraphs $G_{s_1}, \ldots, G_{s_k}$ are pairwise independent. But in fact each of them is also connected since for all $i \in \{1, \ldots, k\}$, $\beta(s_i)$ is a connected dominating set in~$G_{s_i}$. Thus, $G_{s_1}, \ldots, G_{s_k}$ are indeed the connected components of $G_s \setminus \beta(s)$.
	Then, if no vertex rejects at step (iv), it ensures that, for every vertex $u$, the neighborhoods written in $\pieces(u)$ are correct. Moreover, if $u \in \beta(s)$ which is a $P_3$ or a clique of size at most $\sqrt{n}$, the neighborhood of all the vertices in $\beta(s)$ is written in $\pieces(u)$, and if $\beta(s)$ is a clique of size at least $\sqrt{n}$, for every vertex $v \in G_s$, $u$ has a neighbor $w$ such that $\pieces(w)$ contains the neighborhood of $v$.
	
	Assume that $G$ has an induced $P_5$ that we denote by $P$, and let us denote the five vertices of $P$ by $v_1, \ldots, v_5$ (such that $\{v_i,v_j\} \in E$ if and only if $|i-j|=1$). Let $s_1, \ldots, s_5$ be such that $v_i \in \beta(s_i)$ for all $i \in \{1, \ldots, 5\}$. Let us show that at least one vertex among $v_1, \ldots, v_5$ detects $P$ and thus rejects at step (v).
	First, by looking only at the part $\neighbors$ of the certificates, $v_3$ knows the neighborhoods of its neighbors, so in particular those of $v_2$ and $v_4$. Thus, $v_3$ knows all the edges and non-edges of $P$, except the non-edge between $v_1$ and $v_5$. Similarly, $v_2$ knows all the edges and non-edges of $P$ except the edge between $v_4$ and $v_5$, and $v_4$ knows all the edges and non-edges of $P$ except the edge between $v_1$ and $v_2$. This is depicted on Figure~\ref{fig:view of v2, v3, v4}.
	
	\begin{figure}
		\centering
		
		\begin{subfigure}[t]{0.35\textwidth}
			\scalebox{0.8}{
				\begin{tikzpicture}[x=0.75pt,y=0.75pt,yscale=-1,xscale=1]
					
					\draw [color={rgb, 255:red, 208; green, 2; blue, 27 }  ,draw opacity=1 ][line width=1.5]  [dash pattern={on 5.63pt off 4.5pt}]  (110.25,139.25) .. controls (184.5,62) and (307.5,59) .. (383.25,139.25) ;
					\draw  [dash pattern={on 0.84pt off 2.51pt}]  (110.25,139.25) .. controls (177.25,190) and (245.25,192) .. (315,139.25) ;
					\draw  [dash pattern={on 0.84pt off 2.51pt}]  (178.5,139.25) .. controls (245.5,190) and (313.5,192) .. (383.25,139.25) ;
					\draw  [dash pattern={on 0.84pt off 2.51pt}]  (110.25,139.25) .. controls (137.25,109) and (213.25,107) .. (246.75,139.25) ;
					\draw  [dash pattern={on 0.84pt off 2.51pt}]  (246.75,139.25) .. controls (273.75,109) and (349.75,107) .. (383.25,139.25) ;
					\draw  [dash pattern={on 0.84pt off 2.51pt}]  (178.5,139.25) .. controls (205.5,109) and (281.5,107) .. (315,139.25) ;
					\draw    (110.25,139.25) -- (178.5,139.25) ;
					\draw    (178.5,139.25) -- (246.75,139.25) ;
					\draw    (246.75,139.25) -- (315,139.25) ;
					\draw    (315,139.25) -- (383.25,139.25) ;
					\draw  [fill={rgb, 255:red, 255; green, 255; blue, 255 }  ,fill opacity=1 ] (104,139.25) .. controls (104,135.8) and (106.8,133) .. (110.25,133) .. controls (113.7,133) and (116.5,135.8) .. (116.5,139.25) .. controls (116.5,142.7) and (113.7,145.5) .. (110.25,145.5) .. controls (106.8,145.5) and (104,142.7) .. (104,139.25) -- cycle ;
					\draw  [fill={rgb, 255:red, 255; green, 255; blue, 255 }  ,fill opacity=1 ] (172.25,139.25) .. controls (172.25,135.8) and (175.05,133) .. (178.5,133) .. controls (181.95,133) and (184.75,135.8) .. (184.75,139.25) .. controls (184.75,142.7) and (181.95,145.5) .. (178.5,145.5) .. controls (175.05,145.5) and (172.25,142.7) .. (172.25,139.25) -- cycle ;
					\draw  [fill={rgb, 255:red, 245; green, 166; blue, 35 }  ,fill opacity=1 ] (240.5,139.25) .. controls (240.5,135.8) and (243.3,133) .. (246.75,133) .. controls (250.2,133) and (253,135.8) .. (253,139.25) .. controls (253,142.7) and (250.2,145.5) .. (246.75,145.5) .. controls (243.3,145.5) and (240.5,142.7) .. (240.5,139.25) -- cycle ;
					\draw  [fill={rgb, 255:red, 255; green, 255; blue, 255 }  ,fill opacity=1 ] (308.75,139.25) .. controls (308.75,135.8) and (311.55,133) .. (315,133) .. controls (318.45,133) and (321.25,135.8) .. (321.25,139.25) .. controls (321.25,142.7) and (318.45,145.5) .. (315,145.5) .. controls (311.55,145.5) and (308.75,142.7) .. (308.75,139.25) -- cycle ;
					\draw  [fill={rgb, 255:red, 255; green, 255; blue, 255 }  ,fill opacity=1 ] (377,139.25) .. controls (377,135.8) and (379.8,133) .. (383.25,133) .. controls (386.7,133) and (389.5,135.8) .. (389.5,139.25) .. controls (389.5,142.7) and (386.7,145.5) .. (383.25,145.5) .. controls (379.8,145.5) and (377,142.7) .. (377,139.25) -- cycle ;
					
					\draw (240,51) node [anchor=north west][inner sep=0.75pt]   [align=left] {\textcolor[rgb]{0.82,0.01,0.11}{{\large ?}}};
					\draw (112.25,148.9) node [anchor=north west][inner sep=0.75pt]    {$v_{1}$};
					\draw (180.5,148.9) node [anchor=north west][inner sep=0.75pt]    {$v_{2}$};
					\draw (248.75,148.9) node [anchor=north west][inner sep=0.75pt]    {$v_{3}$};
					\draw (317,148.9) node [anchor=north west][inner sep=0.75pt]    {$v_{4}$};
					\draw (385.25,148.9) node [anchor=north west][inner sep=0.75pt]    {$v_{5}$};

				\end{tikzpicture}
			}
			\captionsetup{justification=centering, margin=-1cm}
			\caption{The view of $v_3$.}
		\end{subfigure}
		
		\begin{subfigure}[t]{0.35\textwidth}
			\scalebox{0.8}{
				\begin{tikzpicture}[x=0.75pt,y=0.75pt,yscale=-1,xscale=1]
					
					\draw [color={rgb, 255:red, 0; green, 0; blue, 0 }  ,draw opacity=1 ][line width=0.75]  [dash pattern={on 0.84pt off 2.51pt}]  (110.25,139.25) .. controls (184.5,62) and (307.5,59) .. (383.25,139.25) ;
					\draw  [dash pattern={on 0.84pt off 2.51pt}]  (110.25,139.25) .. controls (177.25,190) and (245.25,192) .. (315,139.25) ;
					\draw  [dash pattern={on 0.84pt off 2.51pt}]  (178.5,139.25) .. controls (245.5,190) and (313.5,192) .. (383.25,139.25) ;
					\draw  [dash pattern={on 0.84pt off 2.51pt}]  (110.25,139.25) .. controls (137.25,109) and (213.25,107) .. (246.75,139.25) ;
					\draw  [dash pattern={on 0.84pt off 2.51pt}]  (246.75,139.25) .. controls (273.75,109) and (349.75,107) .. (383.25,139.25) ;
					\draw  [dash pattern={on 0.84pt off 2.51pt}]  (178.5,139.25) .. controls (205.5,109) and (281.5,107) .. (315,139.25) ;
					\draw    (110.25,139.25) -- (178.5,139.25) ;
					\draw    (178.5,139.25) -- (246.75,139.25) ;
					\draw    (246.75,139.25) -- (315,139.25) ;
					\draw [color={rgb, 255:red, 208; green, 2; blue, 27 }  ,draw opacity=1 ][line width=1.5]  [dash pattern={on 5.63pt off 4.5pt}]  (315,139.25) -- (383.25,139.25) ;
					\draw  [fill={rgb, 255:red, 255; green, 255; blue, 255 }  ,fill opacity=1 ] (104,139.25) .. controls (104,135.8) and (106.8,133) .. (110.25,133) .. controls (113.7,133) and (116.5,135.8) .. (116.5,139.25) .. controls (116.5,142.7) and (113.7,145.5) .. (110.25,145.5) .. controls (106.8,145.5) and (104,142.7) .. (104,139.25) -- cycle ;
					\draw  [fill={rgb, 255:red, 245; green, 166; blue, 35 }  ,fill opacity=1 ] (172.25,139.25) .. controls (172.25,135.8) and (175.05,133) .. (178.5,133) .. controls (181.95,133) and (184.75,135.8) .. (184.75,139.25) .. controls (184.75,142.7) and (181.95,145.5) .. (178.5,145.5) .. controls (175.05,145.5) and (172.25,142.7) .. (172.25,139.25) -- cycle ;
					\draw  [fill={rgb, 255:red, 255; green, 255; blue, 255 }  ,fill opacity=1 ] (240.5,139.25) .. controls (240.5,135.8) and (243.3,133) .. (246.75,133) .. controls (250.2,133) and (253,135.8) .. (253,139.25) .. controls (253,142.7) and (250.2,145.5) .. (246.75,145.5) .. controls (243.3,145.5) and (240.5,142.7) .. (240.5,139.25) -- cycle ;
					\draw  [fill={rgb, 255:red, 255; green, 255; blue, 255 }  ,fill opacity=1 ] (308.75,139.25) .. controls (308.75,135.8) and (311.55,133) .. (315,133) .. controls (318.45,133) and (321.25,135.8) .. (321.25,139.25) .. controls (321.25,142.7) and (318.45,145.5) .. (315,145.5) .. controls (311.55,145.5) and (308.75,142.7) .. (308.75,139.25) -- cycle ;
					\draw  [fill={rgb, 255:red, 255; green, 255; blue, 255 }  ,fill opacity=1 ] (377,139.25) .. controls (377,135.8) and (379.8,133) .. (383.25,133) .. controls (386.7,133) and (389.5,135.8) .. (389.5,139.25) .. controls (389.5,142.7) and (386.7,145.5) .. (383.25,145.5) .. controls (379.8,145.5) and (377,142.7) .. (377,139.25) -- cycle ;
					
					\draw (112.25,148.9) node [anchor=north west][inner sep=0.75pt]    {$v_{1}$};
					\draw (180.5,148.9) node [anchor=north west][inner sep=0.75pt]    {$v_{2}$};
					\draw (248.75,148.9) node [anchor=north west][inner sep=0.75pt]    {$v_{3}$};
					\draw (317,148.9) node [anchor=north west][inner sep=0.75pt]    {$v_{4}$};
					\draw (385.25,148.9) node [anchor=north west][inner sep=0.75pt]    {$v_{5}$};
					\draw (344,116) node [anchor=north west][inner sep=0.75pt]   [align=left] {\textcolor[rgb]{0.82,0.01,0.11}{{\large ?}}};

				\end{tikzpicture}
				
			}
			\captionsetup{justification=centering, margin=-1cm}
			\caption{The view of $v_2$.}
		\end{subfigure}
		\hspace{2cm}
		\begin{subfigure}[t]{0.35\textwidth}
			\scalebox{0.8}{
				\begin{tikzpicture}[x=0.75pt,y=0.75pt,yscale=-1,xscale=1]
					
					\draw [color={rgb, 255:red, 0; green, 0; blue, 0 }  ,draw opacity=1 ][line width=0.75]  [dash pattern={on 0.84pt off 2.51pt}]  (110.25,139.25) .. controls (184.5,62) and (307.5,59) .. (383.25,139.25) ;
					\draw  [dash pattern={on 0.84pt off 2.51pt}]  (110.25,139.25) .. controls (177.25,190) and (245.25,192) .. (315,139.25) ;
					\draw  [dash pattern={on 0.84pt off 2.51pt}]  (178.5,139.25) .. controls (245.5,190) and (313.5,192) .. (383.25,139.25) ;
					\draw  [dash pattern={on 0.84pt off 2.51pt}]  (110.25,139.25) .. controls (137.25,109) and (213.25,107) .. (246.75,139.25) ;
					\draw  [dash pattern={on 0.84pt off 2.51pt}]  (246.75,139.25) .. controls (273.75,109) and (349.75,107) .. (383.25,139.25) ;
					\draw  [dash pattern={on 0.84pt off 2.51pt}]  (178.5,139.25) .. controls (205.5,109) and (281.5,107) .. (315,139.25) ;
					\draw    (315,139.25) -- (383.25,139.25) ;
					\draw    (178.5,139.25) -- (246.75,139.25) ;
					\draw    (246.75,139.25) -- (315,139.25) ;
					\draw [color={rgb, 255:red, 208; green, 2; blue, 27 }  ,draw opacity=1 ][line width=1.5]  [dash pattern={on 5.63pt off 4.5pt}]  (110.25,139.25) -- (178.5,139.25) ;
					\draw  [fill={rgb, 255:red, 255; green, 255; blue, 255 }  ,fill opacity=1 ] (104,139.25) .. controls (104,135.8) and (106.8,133) .. (110.25,133) .. controls (113.7,133) and (116.5,135.8) .. (116.5,139.25) .. controls (116.5,142.7) and (113.7,145.5) .. (110.25,145.5) .. controls (106.8,145.5) and (104,142.7) .. (104,139.25) -- cycle ;
					\draw  [fill={rgb, 255:red, 255; green, 255; blue, 255 }  ,fill opacity=1 ] (172.25,139.25) .. controls (172.25,135.8) and (175.05,133) .. (178.5,133) .. controls (181.95,133) and (184.75,135.8) .. (184.75,139.25) .. controls (184.75,142.7) and (181.95,145.5) .. (178.5,145.5) .. controls (175.05,145.5) and (172.25,142.7) .. (172.25,139.25) -- cycle ;
					\draw  [fill={rgb, 255:red, 255; green, 255; blue, 255 }  ,fill opacity=1 ] (240.5,139.25) .. controls (240.5,135.8) and (243.3,133) .. (246.75,133) .. controls (250.2,133) and (253,135.8) .. (253,139.25) .. controls (253,142.7) and (250.2,145.5) .. (246.75,145.5) .. controls (243.3,145.5) and (240.5,142.7) .. (240.5,139.25) -- cycle ;
					\draw  [fill={rgb, 255:red, 245; green, 166; blue, 35 }  ,fill opacity=1 ] (308.75,139.25) .. controls (308.75,135.8) and (311.55,133) .. (315,133) .. controls (318.45,133) and (321.25,135.8) .. (321.25,139.25) .. controls (321.25,142.7) and (318.45,145.5) .. (315,145.5) .. controls (311.55,145.5) and (308.75,142.7) .. (308.75,139.25) -- cycle ;
					\draw  [fill={rgb, 255:red, 255; green, 255; blue, 255 }  ,fill opacity=1 ] (377,139.25) .. controls (377,135.8) and (379.8,133) .. (383.25,133) .. controls (386.7,133) and (389.5,135.8) .. (389.5,139.25) .. controls (389.5,142.7) and (386.7,145.5) .. (383.25,145.5) .. controls (379.8,145.5) and (377,142.7) .. (377,139.25) -- cycle ;
					
					\draw (112.25,148.9) node [anchor=north west][inner sep=0.75pt]    {$v_{1}$};
					\draw (180.5,148.9) node [anchor=north west][inner sep=0.75pt]    {$v_{2}$};
					\draw (248.75,148.9) node [anchor=north west][inner sep=0.75pt]    {$v_{3}$};
					\draw (317,148.9) node [anchor=north west][inner sep=0.75pt]    {$v_{4}$};
					\draw (385.25,148.9) node [anchor=north west][inner sep=0.75pt]    {$v_{5}$};
					\draw (145,116) node [anchor=north west][inner sep=0.75pt]   [align=left] {\textcolor[rgb]{0.82,0.01,0.11}{{\large ?}}};

				\end{tikzpicture}
				
			}
			\captionsetup{justification=centering, margin=-1cm}
			\caption{The view of $v_4$.}
		\end{subfigure}
		
		\caption{The edges and non-edges seen by $v_3$, $v_2$ and $v_4$ in $P$, using just the part of $\neighbors$ of the certificates. The edges are represented by the normal edges, the non-edges are represented by the dotted ones, and the red dashed edges represent the unknown edges or non-edges.}
		\label{fig:view of v2, v3, v4}
	\end{figure}
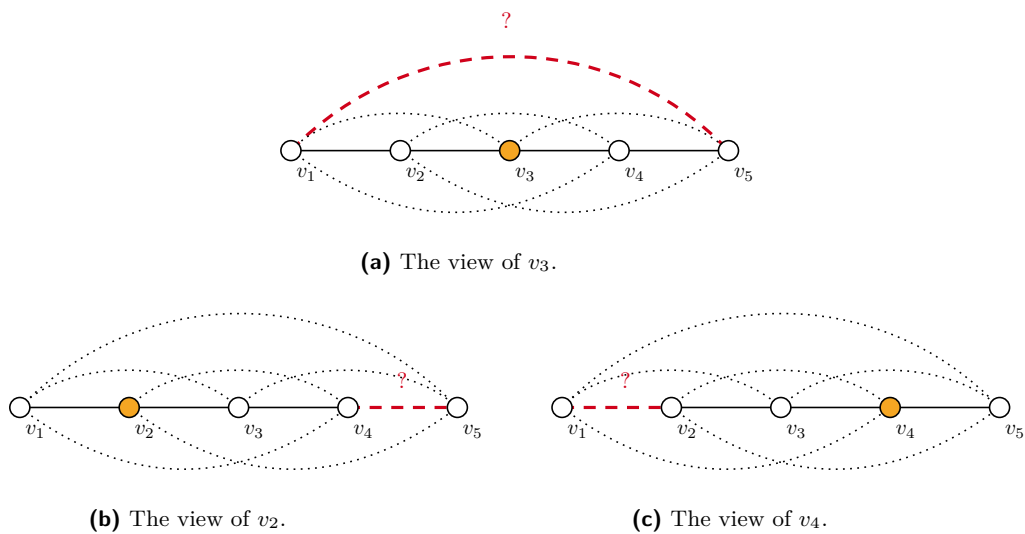
	
	By Lemma~\ref{claim:edges_ancestors}, there exists \mbox{$i_0 \in \{1, \ldots, 5\}$} such that $s_{i_0}$ is an ancestor of $s_i$ for every $i \in \{1, \ldots 5\}$.
	If $\beta({s_{i_0}})$ is a clique of size at least $\sqrt{n}$, then $v_{i_0}$ detects $P$ and rejects (indeed, $P$ is entirely contained in $G_{s_{i_0}}$, and since $\beta(s_{i_0})$ is a clique of size at least $\sqrt{n}$, $v_{i_0}$ knows all the neighborhoods of vertices in $G_{s_{i_0}}$ by looking at the part $\pieces$ of the certificates of vertices in~$\beta({s_{i_0}})$).
	Thus, we can assume that $\beta(s_{i_0})$ is either a $P_3$ or a clique of size at most $\sqrt{n}$ in $G$. In both cases, the neighborhood of every vertex in $\beta(s_{i_0})$ is written in the certificate of all vertices of $\beta(s_{i_0})$.
	For every $i \in \{1, \ldots, 5\}$, $v_i$ is a neighbor of some vertex in $\beta(s_{i_0})$ since $\beta(s_{i_0})$ is a dominating set in $G_{s_{i_0}}$. Thus $v_i$ knows the neighborhood of $v_{i_0}$ (which is written in the certificate of all the vertices of $\beta(s_{i_0})$). In particular, if $i_0 \in \{1,2\}$, then $v_4$ knows the neighborhood of $v_1$ or $v_2$, so $v_4$ detects $P$ entirely (since the only possible unknown edge of $P$ was $v_1v_2$, see Figure~\ref{fig:view of v2, v3, v4}) and rejects. Similarly, if $i_0 \in \{4,5\}$, $v_2$ detects $P$ entirely and rejects.
	Thus, we can assume that $i_0 = 3$. So $v_1, v_2, v_4, v_5$ all know the adjacency of~$v_3$.

	If $s_1$ and $s_5$ are not in the same branch of~$T$, $v_3$ knows the non-edge between $v_1$ and $v_5$, so it detects $P$ and rejects. Thus, we can assume that $s_1$ and $s_5$ are in the same branch.
	By symmetry, assume that $s_1$ is a descendant of $s_5$ (which itself is a descendant of $s_3$). There are finally  two cases to conclude:
	\begin{enumerate}
		\item If $\beta(s_5)$ is a $P_3$ or a clique of size at most~$\sqrt{n}$, then $v_1$ detects the path. Indeed, since $s_1$ is a descendant of $s_5$, $v_1$ knows the adjacency of $v_5$. Since $v_1$ also knows the adjacency of~$v_2$ (written in $\neighbors(v_2)$) and~$v_3$ (because $i_0=3$), it sees all the edges and non-edges of~$P$, so it detects~$P$ and rejects.
		
		\item If $\beta(s_5)$ is a clique of size at least~$\sqrt{n}$, then $v_5$ detects the path. Indeed, since $s_1$ is a descendant of $s_5$, $v_5$ knows the adjacency of $v_1$. Since $v_5$ also knows the adjacency of~$v_4$ (written in $\neighbors(v_4)$) and~$v_3$ (because $i_0=3$), it sees all the edges and non-edges of~$P$, so it detects~$P$ and rejects.
	\end{enumerate}

	Thus, in all cases, at least one vertex detects $P$ and rejects, which concludes the proof.

	\bibliographystyle{plainurl}
	\bibliography{bibli}

\begin{thebibliography}{10}

\bibitem{bacso1990dominating}
G{\'a}bor Bacs{\'o} and Zs. Tuza.
\newblock Dominating cliques in $p_5$-free graphs.
\newblock {\em Periodica Mathematica Hungarica}, 21(4):303--308, 1990.

\bibitem{BousquetCFPZ24+}
Nicolas Bousquet, Linda Cook, Laurent Feuilloley, Th{\'{e}}o Pierron, and
  S{\'{e}}bastien Zeitoun.
\newblock Local certification of forbidden subgraphs.
\newblock {\em CoRR}, abs/2402.12148, 2024.
\newblock URL: \url{https://doi.org/10.48550/arXiv.2402.12148}, \href
  {http://arxiv.org/abs/2402.12148} {\path{arXiv:2402.12148}}, \href
  {https://doi.org/10.48550/ARXIV.2402.12148}
  {\path{doi:10.48550/ARXIV.2402.12148}}.

\bibitem{BousquetEFZ24}
Nicolas Bousquet, Louis Esperet, Laurent Feuilloley, and S{\'{e}}bastien
  Zeitoun.
\newblock Renaming in distributed certification.
\newblock {\em CoRR}, abs/2409.15404, 2024.
\newblock \href {http://arxiv.org/abs/2409.15404} {\path{arXiv:2409.15404}},
  \href {https://doi.org/10.48550/ARXIV.2409.15404}
  {\path{doi:10.48550/ARXIV.2409.15404}}.

\bibitem{BousquetFP21}
Nicolas Bousquet, Laurent Feuilloley, and Th{\'{e}}o Pierron.
\newblock Local certification of graph decompositions and applications to
  minor-free classes.
\newblock In Quentin Bramas, Vincent Gramoli, and Alessia Milani, editors, {\em
  25th International Conference on Principles of Distributed Systems, {OPODIS}
  2021, December 13-15, 2021, Strasbourg, France}, volume 217 of {\em LIPIcs},
  pages 22:1--22:17. Schloss Dagstuhl - Leibniz-Zentrum f{\"{u}}r Informatik,
  2021.
\newblock URL: \url{https://doi.org/10.4230/LIPIcs.OPODIS.2021.22}, \href
  {https://doi.org/10.4230/LIPICS.OPODIS.2021.22}
  {\path{doi:10.4230/LIPICS.OPODIS.2021.22}}.

\bibitem{CambyS16}
Eglantine Camby and Oliver Schaudt.
\newblock A new characterization of $p_k$-free graphs.
\newblock {\em Algorithmica}, 75(1):205--217, 2016.
\newblock URL: \url{https://doi.org/10.1007/s00453-015-9989-6}, \href
  {https://doi.org/10.1007/S00453-015-9989-6}
  {\path{doi:10.1007/S00453-015-9989-6}}.

\bibitem{CensorHillel22}
Keren Censor{-}Hillel.
\newblock Distributed subgraph finding: Progress and challenges.
\newblock {\em CoRR}, abs/2203.06597, 2022.
\newblock URL: \url{https://doi.org/10.48550/arXiv.2203.06597}, \href
  {http://arxiv.org/abs/2203.06597} {\path{arXiv:2203.06597}}, \href
  {https://doi.org/10.48550/ARXIV.2203.06597}
  {\path{doi:10.48550/ARXIV.2203.06597}}.

\bibitem{ChaniotisCHS24}
Aristotelis Chaniotis, Linda Cook, Sepehr Hajebi, and Sophie Spirkl.
\newblock Personal communication, 2024.

\bibitem{DruckerKO13}
Andrew Drucker, Fabian Kuhn, and Rotem Oshman.
\newblock On the power of the congested clique model.
\newblock In Magn{\'{u}}s~M. Halld{\'{o}}rsson and Shlomi Dolev, editors, {\em
  {ACM} Symposium on Principles of Distributed Computing, {PODC} '14, Paris,
  France, July 15-18, 2014}, pages 367--376. {ACM}, 2014.
\newblock \href {https://doi.org/10.1145/2611462.2611493}
  {\path{doi:10.1145/2611462.2611493}}.

\bibitem{EsperetL22}
Louis Esperet and Benjamin L{\'{e}}v{\^{e}}que.
\newblock Local certification of graphs on surfaces.
\newblock {\em Theor. Comput. Sci.}, 909:68--75, 2022.
\newblock URL: \url{https://doi.org/10.1016/j.tcs.2022.01.023}, \href
  {https://doi.org/10.1016/J.TCS.2022.01.023}
  {\path{doi:10.1016/J.TCS.2022.01.023}}.

\bibitem{Feuilloley21}
Laurent Feuilloley.
\newblock Introduction to local certification.
\newblock {\em Discret. Math. Theor. Comput. Sci.}, 23(3), 2021.
\newblock URL: \url{https://doi.org/10.46298/dmtcs.6280}, \href
  {https://doi.org/10.46298/DMTCS.6280} {\path{doi:10.46298/DMTCS.6280}}.

\bibitem{FeuilloleyBP22}
Laurent Feuilloley, Nicolas Bousquet, and Th{\'{e}}o Pierron.
\newblock What can be certified compactly? compact local certification of {MSO}
  properties in tree-like graphs.
\newblock In Alessia Milani and Philipp Woelfel, editors, {\em {PODC} '22:
  {ACM} Symposium on Principles of Distributed Computing, Salerno, Italy, July
  25 - 29, 2022}, pages 131--140. {ACM}, 2022.
\newblock \href {https://doi.org/10.1145/3519270.3538416}
  {\path{doi:10.1145/3519270.3538416}}.

\bibitem{FeuilloleyFMRRT21}
Laurent Feuilloley, Pierre Fraigniaud, Pedro Montealegre, Ivan Rapaport,
  {\'{E}}ric R{\'{e}}mila, and Ioan Todinca.
\newblock Compact distributed certification of planar graphs.
\newblock {\em Algorithmica}, 83(7):2215--2244, 2021.
\newblock URL: \url{https://doi.org/10.1007/s00453-021-00823-w}, \href
  {https://doi.org/10.1007/S00453-021-00823-W}
  {\path{doi:10.1007/S00453-021-00823-W}}.

\bibitem{FraigniaudM0RT23}
Pierre Fraigniaud, Fr{\'{e}}d{\'{e}}ric Mazoit, Pedro Montealegre, Ivan
  Rapaport, and Ioan Todinca.
\newblock Distributed certification for classes of dense graphs.
\newblock In Rotem Oshman, editor, {\em 37th International Symposium on
  Distributed Computing, {DISC} 2023, October 10-12, 2023, L'Aquila, Italy},
  volume 281 of {\em LIPIcs}, pages 20:1--20:17. Schloss Dagstuhl -
  Leibniz-Zentrum f{\"{u}}r Informatik, 2023.
\newblock URL: \url{https://doi.org/10.4230/LIPIcs.DISC.2023.20}, \href
  {https://doi.org/10.4230/LIPICS.DISC.2023.20}
  {\path{doi:10.4230/LIPICS.DISC.2023.20}}.

\bibitem{FraigniaudMRT24}
Pierre Fraigniaud, Pedro Montealegre, Ivan Rapaport, and Ioan Todinca.
\newblock A meta-theorem for distributed certification.
\newblock {\em Algorithmica}, 86(2):585--612, 2024.
\newblock URL: \url{https://doi.org/10.1007/s00453-023-01185-1}, \href
  {https://doi.org/10.1007/S00453-023-01185-1}
  {\path{doi:10.1007/S00453-023-01185-1}}.

\bibitem{GoosS16}
Mika G{\"{o}}{\"{o}}s and Jukka Suomela.
\newblock Locally checkable proofs in distributed computing.
\newblock {\em Theory Comput.}, 12(1):1--33, 2016.
\newblock URL: \url{https://doi.org/10.4086/toc.2016.v012a019}, \href
  {https://doi.org/10.4086/TOC.2016.V012A019}
  {\path{doi:10.4086/TOC.2016.V012A019}}.

\end{thebibliography}
	
\end{document}